	\documentclass[11pt]{article}
	\usepackage{geometry}
	\usepackage{array}
	\usepackage{amsthm,amssymb,mathtools}
	\usepackage[dvipsnames]{xcolor}
	\usepackage{cite}
	\usepackage{fancyvrb}
	\usepackage{graphicx}
	\usepackage{hyperref}
	\usepackage{epstopdf}
	
	\newtheorem{theorem}{Theorem}
	\theoremstyle{definition}
	\newtheorem{definition}{Definition}[section]
\title{
	A superstatistical formulation of complexity measures
	}
\author{
	Jes\'us Fuentes \& Octavio Obreg\'on
	}
\date{}
\begin{document}
\maketitle
%	----------------------------------------------------------------------/
%	//////////////////////////////////////////////////////////////////////
%	----------------------------------------------------------------------/
\begin{abstract}

It is discussed how the superstatistical formulation of effective Boltzmann factors can be related to the concept of Kolmogorov complexity, generating an infinite set of complexity measures (CMs) for quantifying information. At this level the information is treated according to its background, which means that the CM depends on the inherent attributes of the information scenario. While the basic Boltzmann factor directly produces the standard complexity measure (SCM), it succeeds in the description of large-scale scenarios where the data components are not interrelated with themselves, thus adopting the behaviour of a gas. What happens in scenarios in which the presence of sources and sinks of information cannot be neglected, needs of a CM other than the one produced by the ordinary Boltzmann factor. We introduce a set of flexible CMs, without free parameters, that converge asymptotically to the Kolmogorov complexity, but also quantify the information in scenarios with a reasonable small density of states. We prove that these CMs are obtained from a generalised relative entropy and we suggest why such measures are the only compatible generalisations of the SCM.
 
\end{abstract}
%	----------------------------------------------------------------------/
%	//////////////////////////////////////////////////////////////////////
%	----------------------------------------------------------------------/
\section{Introduction}
\label{intro}

Suppose we are given the binary strings
\begin{equation*}
\begin{aligned}
	&\verb|10101010101010101010|\\
	&\verb|00100110011110111001|
\end{aligned}
\end{equation*}
and then we are asked how random they are. Without further analysis, the first string is simply a sequence of ten times \verb|10|, hence to quantify its randomness is rather meaningless. On the other hand, although the second string looks more complicated than the first one, it merely corresponds to the first twenty bits of the decimal part of $\pi$, therefore its description can be put into unchallenging terms as well. The more economic description we can glimpse to resemble these binary objects the closer we are to their actual {\it Kolmogorov complexity} \cite{kolmogorov, solomonoff0, chaitin1},  which is, roughly speaking, a measure of randomness that looks for the shortest possible program that halts and delivers the bit string in question, the shorter the program, the lesser random is its output.

From a statistical viewpoint, however, the two bit strings above have the same probability, $2^{-20}$, of being picked from the whole set of binary sequences of twenty bits. Even though, it is also interesting to ask for the probability with which we can pull a specific program, from a collection of programs, that prints out a bit string like the ones above. To be more clear, suppose that a set of programs are recorded in the memory of a computer but we ignore which state it is in. Suppose the only thing we know is the expected value of the length of the program the computer outperforms. Of course we can improve our chances of guessing if we are given more information. Suppose we additionally know the outcome itself, hence we are led again to the Kolmogorov complexity of the program involved. Yet we have now an additional component. From the Gibbs formulation of ensemble theory, we know that the probability that the computer in question is in state (program) $x$ is quantified by $p(x)$, as long as this distribution maximises a suitable entropy measure subject to some specific constraints that we shall discuss soon.

In this sense we are associating two different concepts that account for information although each one by separate paths and with distinct interpretations: (1) The entropy, a functional which depends on a probability distribution (and sometimes on parameters) and, (2) The Kolmogorov complexity or {\it algorithmic entropy}, a descriptive quantity that depends on the object itself. The relation between these two concepts has been matter of a substantial discussion \cite{chaitin, zvonkin, tadaki, baez, vitanyi, galatolo}, but here we want to adopt a different approach.

In this work we introduce the notion of superstatistics \cite{beck} as an option to extend the algorithmic information theory to an infinite number of CMs. The proposal is rather simple. It consists of relating the concept of generalised entropy to their corresponding CM. Nonetheless, the entropy as a measure of information must satisfy the criterion of stability \cite{lesche} to become eligible as an adequate tool for such purposes. Which reduces the whole universe of generalised entropies to a small set. To our knowledge, there are only two entropies independent from parameters that generalise the Shannon's entropy and  simultaneously fulfil the condition of stability \cite{paper} ---in consequence these are good candidates to deal with information. For this reason, we support our following discussion in such entropy measures, namely $H_+(X)$ and $H_-(X)$, and their corresponding CMs, from now on identified as $K_+(X)$ and $K_-(X)$.

Interestingly, we also arrive to $K_+(X)$ and $K_-(X)$ from the relative versions of $H_+(X)$ and $H_-(X)$, in turn interpreted as general representations of the Kullback-Leiber divergence \cite{kullback} for the Shannon's entropy.

Our discussion is organised as follows. In Sec. \ref{superstatistics} we introduce the superstatistics framework and establish the concept of entropy on this context. We also formulate the connection between the generalised entropy and its respective CM. In Sec. \ref{entropies} we analyse statistically the consequences driven by a generalised Kolmogorov complexity in terms of the entropies $H_+(X)$ and $H_-(X)$. As a result we find the effective complexities $K_+(X)$ and $K_-(X)$. Finally in Sec. \ref{conclusions} we summarise our results to conclude our work.

%	----------------------------------------------------------------------/
%	//////////////////////////////////////////////////////////////////////
%	----------------------------------------------------------------------/
\section{Algorithmic superstatistics}
\label{superstatistics}

The superstatistics approach \cite{beck} handles non-equilibrium macroscopic systems segmented into cells that manifest asymptotically stationary states with a spatiotemporally fluctuating intensive quantity, typically the inverse temperature $\beta$, which hardly varies on a long time scale. To each cell there is assigned a particular $\beta$ distributed according to a piecewise continuous, normalisable probability density $f(\beta)$. Inside each region $\beta$ is approximately constant and therefore there is essentially local equilibrium. At a global level, the system lies in a state out of equilibrium but enough isolated from external upheavals, thus it only behaves slightly deviated from equilibrium. When the entire fluctuations are averaged, inasmuch as the sum converges, an effective Boltzmann factor is obtained giving rise to general statistics. Even more, given that the method relies on those normalisable distributions $f(\beta)$ we can speak of an infinite set of possible generalised statistics.

To knit the pieces together, we are to think of these cells as individual programs $x$ that cast an outcome and halt, suggesting that the entire system can be thought of as a universal computer $U$ of general purpose ---indeed in a real scenario, the tasks that a computer outperforms are upshots of a big collection of different recipes stored in it. Accordingly, if we denote with $|x|$ the length of the program $x$, hence
\begin{equation}
\label{BE}
B(|x|)=\int\limits_0^\infty d\beta\,f(\beta)e^{-\beta |x|}, \quad \beta>0,
\end{equation}
is the effective Boltzmann factor related to the universal computer $U$. In the case that all the cells have the same $\beta$, then the system can be considered as a single cell, in that case $f(\beta)=\delta(\beta-\beta_0)$, and therefore we have
\begin{equation*}
B(|x|) = e^{-\beta |x|},
\end{equation*}
which is the Boltzmann factor in the orthodox picture of statistical mechanics. 

Yet the generalised statistics have to be normalisable over the whole domain of program lengths, for that reason the integral (partition function $Z$)
\begin{equation*}
\int\limits_0^\infty d|x|\,B(|x|)
\end{equation*}
must converge. Nevertheless, the partition function is not always computable, for instance, in the particular case $B(|x|) = e^{-\gamma |x|}$ ---which is the ordinary Boltzmann factor--- the integral exists for $\beta \geq \ln 2$ although it is uncomputable and partially random for $\ln2$ as proved by Tadaki \cite{tadaki2}. Exploring the (un) computability of $Z$ in the general case, constitutes a mathematical challenge beyond the scope of this work, and we shall omit such discussion.

There might be circumstances in which we are given an effective Boltzmann factor but instead ignore the generating distribution, in that case we can compute it by reversing \eqref{BE} as
\begin{equation}
\label{distribution}
f(\beta) = \mathrm{Re}\left[\frac{1}{2\pi i}\int\limits_{|x|-i\infty}^{|x|+i\infty}d|x'|\, B(|x'|)e^{\beta |x'|}\right],
\end{equation}
we shall remark, however, that under this integral transformation, $f(\beta)$ is not univocally determined by the effective Boltzmann factor $B(\|x\|)$. 

We are now to connect these concepts with the algorithmic entropy. As we have previously discussed, the core of the superstatistics approach is the effective Boltzmann factor $B(|x|)$ whose construction rests on a probability distribution well-nigh customised for a system of particular characteristics. In practice, this serves as a pivot to sprout general statistics in which the related fundamental quantities (such as entropy, free energy, etc.) must be rewritten in the frame of the new scheme.

Of special interest is the general expression for entropy in superstatistics. Before giving its formal definition, it is important to state that we are to consider only entropy measures of the form
\begin{equation}
\label{entropy}
H(X) = \sum_{x \in X} h(p(x)),
\end{equation}
where $h(p(x))$ is known as the entropic form, abbreviated as $h(x)$, and $p(x)$ is the probability distribution according which the program $x$ is distributed over a set $X$.

As it has been nicely shown in \cite{souza}, one can compute the entropic form as
\begin{equation}
\label{entropic}
h(x) = \int\limits_0^x dy\, \frac{\alpha+|y|}{1-|y|/|y^*|}, 
\end{equation}
where the length $|\cdot|$ is the inverse function of \eqref{BE} with minimum value $|\cdot^*|$, if any, and $\alpha$ is a constant that has to be identified from the condition $h(1)=0$. It is worth mentioning that the latter formula comes from a MaxEnt programme, in that regard one should expect that if $|x|$ is such that satisfies formulae \eqref{BE} and \eqref{distribution} then it maximises \eqref{entropic}.

To proceed, we need to establish the following assumptions: The integral \eqref{entropic} exists and can be computed; and the entropy \eqref{entropy} can be written according to the ansatz
\begin{equation*}
H(X) = -\sum_{x\in X} p(x) \Lambda(p(x)),
\end{equation*}
where $\Lambda(x)$ is an effective logarithm whose corresponding inverse function is an effective exponential $\epsilon(x)$ such that $\Lambda(\epsilon(x))=\epsilon(\Lambda(x))=x$, subject to the conditions $\Lambda(1)=0$ and $\Lambda'(1)=1$, see Appendix \ref{app}. Along this work, all logarithms either effective or not are base 2, without further notation unless a different thing needs to be specified.

%	----------------------------------------------------------------------/
% 	**********************************************************************
%	----------------------------------------------------------------------/
\subsection{A superstatistical measure of complexity}
\label{superentropy}

In turn, we shall examine some of the aspects conveyed by the superstatistical formulation of entropy and its consequences in the measurement of complexity. To this aim, henceforth we are to consider only recursive probability distributions, that is, distributions computable with a Turing machine $U$, that we assume is running over a prefix-free domain $X$. Recall that a set of strings $X$ is prefix-free if no string $x\in X$ is prefix of another string $x'\in X$.

There is a number of approaches \cite{chaitin2001, lempel} to measure the randomness of an object $x$. But in particular there is an uncomputable measure, although conceptually riveting, known as {\it Kolmogorov complexity} \cite{kolmogorov, solomonoff0, chaitin1}. It is formally defined as follows:

\begin{definition}{Kolmogorov complexity.}\label{kolmogorov} Let $U$ be a prefix-free Turing Machine, the complexity of the string $y$ with respect to $U$ is determined as
\begin{equation*}
K_U(y) = \min_x\{|x|: U(x)=y\},
\end{equation*}
that is the minimum possible length over all programs $x$ with the halt property, which outcome is $y$.
\end{definition}

As it was argued in Sec. \ref{intro}, the quantity $K_U(y)$ has an intuitive but profound meaning. For a person describing the recipe for beef goulash to another person such that this one cannot make a different interpretation of the directions  for the correct realisation of such meal, then the number of bits in that communication constitutes an upper bound on $K_U(y)$.

Nonetheless, we want now to pursue a partially different approach regarding Def. \ref{kolmogorov}. Instead of considering the complexity associated with a program $x$, we are interested in the probability $p(x)$ associated with that program, i.e. we look for the way this (minimum-length) program is distributed over the domain $X$ of programs that can achieve a specific outcome $y$. Thus rather than $K_U(y)$, from now on we refer to this quantity as $K(X) \equiv K_U(p(x))$.

Both $K_U(y)$ and $K(X)$ are measures of information. The first one strictly arises from combinatorial arguments, while the latter is a purely statistical measure of complexity that can be thought of as the average rate at which information is extracted from a combinatorial trial. Since the second measure of complexity depends on how the program $x$ is distributed according to the law $p(x)$, there can be a number of ways in which such probability distribution can be maximised, depending on the entropy and maybe on some inherent constraints that are typically related with the special characteristics of a system.

Even more, parallel to the discussion in \cite{campbell0}, the complexity $K(X)$ can be weighted according to a function $\varphi(K(X))$ that quantifies the cost of managing specific rates of complexity. The Nagumo-Kolmogorov function $\varphi$ is not arbitrary but depends on the entropy functional that maximises the distribution $p(x)$, such that
\begin{equation*}
0\leq \varphi^{-1}\left(\sum_{x\in X} p(x) \varphi(K(x))\right),
\end{equation*}
then, for any superstatistical entropy of the form \eqref{entropic}, one obtains an effective coding theorem (formally discussed in Sec. \ref{codingcomplexity})
\begin{equation}
\label{codtheo}
\sum_{x\in X} h(x) \leq \varphi^{-1}\left(\sum_{x\in X} p(x) \varphi(K(x))\right),
\end{equation}
where, as shown in \cite{levin, kirchherr}, the complexity $K(x)$ is indeed related to the Chaitin formulation of complexity \cite{chaitin2001}, according to the formula $K(x) = -\Lambda(m(x)) + O(1)$, with
\begin{equation}
\label{chaitin}
m(y) = \left\{\sum_{x\in X} 2^{-|x|}:U(x)=y\right\},
\end{equation}
that is, $m$ is the probability that $y$ is the output of a universal Turing machine $U$ running over the programs $X$. The summation over the whole set of programs in \eqref{chaitin} can be simplified by reasoning that there is only one program $x$ in $X$ which outcome is $y$, to see this imagine $n$ programs $x_1,\ldots,x_n$, all having the same outcome $y$, yet we are only interested in the shortest one given that $K$ regards the minimum possible description of $y$. Moreover, in the hypothetical case that $X$ contains $n$ minimal programs $x=x_1=\cdots=x_n$, all printing the same output $y$, then $|x|>\ln n$. Under these considerations then \eqref{codtheo} is equivalently expressed as
\begin{equation*}
\sum_{x\in X} h(x) \leq \varphi^{-1}\left(\sum_{x\in X} p(x) \varphi(|x| + O(1))\right),
\end{equation*}
consequently, we are entitled to write the following relation
\begin{equation}
\label{theorem}
0 \leq  \varphi^{-1}\left(\sum_{x\in X} p(x) \varphi(K(x))\right) - \sum_{x\in X} h(x) \leq \varphi(K(X)),
\end{equation}
this expression constitutes a theorem, and tells us that the entropy and the complexity, as measures of information, are truly connected if the complexity is treated from a statistical viewpoint. Not only that, the latter relation comprises the coding theorem formulated by Shannon \cite{shannon}, indicating us that the entropy provides the minimum rate at which the complexity can be expressed. Some of the consequences conveyed by the relation \eqref{theorem}, in terms of generalised entropies, shall be surveyed in the following section.

%	----------------------------------------------------------------------/
%	//////////////////////////////////////////////////////////////////////
%	----------------------------------------------------------------------/
\section{Effective algorithmic entropies}
\label{entropies}

As it has been discussed earlier, the connection between superstatistics and the algorithmic formulation of information theory has the advantage of quantifying data in generalised scenarios with non-stationary information fluxes.  

That is the case of a theoretical computer partitioned into modules (or cells) that can be randomly accessed depending on the tasks that shall be executed. Each of theses modules will have an individual expected algorithmic-running time, such that the global execution time is distributed over the whole collection of modules that participate in a specific task. However, the way in which the global algorithmic-running time is distributed over the modules may vary according to the computer's configuration. For instance, if the collection of modules are completely independent from each other, a probability distribution maximised through the Shannon's entropy may describe such scenario. Although this is merely a special situation, in that many applications lead to non-standard distributions, and the necessity of a generalised entropy becomes evident.

This happens, in computing scenarios in which a module-module interaction becomes negligible in the presence of a high density of programs (states, in physics). While at this level the Shannon's entropy suffices to maximise the underlying distribution, the story is not the same in case the density of programs is reasonable small, given that the interactions cannot be entirely disregarded. In this case, the Shannon's entropy needs of nonlinear, correction terms to account for such interactions.

From the superstatistics viewpoint, it has been studied in \cite{oo10,oo18} that this kind of systems could be well characterised via Gamma-like distributions conveying shape parameters that {\it a posteriori} can be identified with a generic probability distribution $p(x)$ as
\begin{equation*}
f_{p(x)}^\pm(\beta) = \frac{1}{\beta_0p(x)\Gamma\left(\frac{1}{p(x)}\right)} \left(\frac{\beta}{\beta_0}\frac{1}{p(x)}\right)^{\frac{\pm1 - p(x)}{p(x)}} \exp\left(-\frac{\beta}{\beta_0p(x)}\right),
\end{equation*}
now the integral in \eqref{BE} can be outperformed with these distributions to  obtain a pair of effective Boltzmann factors
\begin{equation}
B_{p(x)}^\pm(|x|) = (1 \pm p(x)\beta_0|x|)^{\mp\frac{1}{p(x)}}.
\end{equation}

Following the steps described in Sec. \ref{superstatistics}, it can be shown that substituting the inverse of $B_{p_i}^\pm(|x|)$ into formulae \eqref{entropic} and \eqref{entropy}, in that order, one obtains the two entropy measures
\begin{equation}
\label{masmenos}
H_+(X)= -\sum_{x\in X} p(x)\ln_{+}(p(x)), \quad H_-(X) = -\sum_{x\in X} p(x)\ln_{-}(p(x)),
\end{equation}
where the effective logarithms define as $\ln_+(\xi)\equiv-(1-\xi^\xi)/\xi$ and $\ln_-(\xi)\equiv-(\xi^{-\xi}-1)/\xi$ for $\xi\in(0,1]$. 

These two entropy measures converge to the Shannon entropy in the regime of low probabilities or high density of programs \cite{paper}, which results evident from the series representations 
\begin{equation}
\label{series}
H_+(X) = -\sum_{x\in X} \sum_{k\in\mathbb{N}}\frac{[p(x) \ln p(x)]^k}{k!}, \quad H_-(X) = -\sum_{x\in X} \sum_{k\in\mathbb{N}}(-1)^{k+1}\frac{[p(x) \ln p(x)]^k}{k!}, 
\end{equation}
remarking that in general there is a region identified with physical phenomena slightly out of equilibrium \cite{fp,gil17,huicho}. We are to focus our analysis on the latter region, since the limiting case of Shannon is already know.  

Unlike other non-extensive entropies the functionals \eqref{masmenos} do not depend on free parameters but only on the probability distribution. Actually, the lack of parameters and their asymptotical behaviour grant both entropies of full stability, see \cite{paper}, which is an essential attribute for managing information either statistical or algorithmically, as we are to show now.

%	----------------------------------------------------------------------/
%	//////////////////////////////////////////////////////////////////////
%	----------------------------------------------------------------------/
\subsection{From coding theorems to complexity measures}
\label{codingcomplexity}

The superstatistical formulation of entropy can account for systems out of equilibrium. We have also discussed how the entropy and the Kolmogorov complexity as measures of information relate each other through \eqref{theorem}. In turn, we are to translate those arguments to the entropies \eqref{masmenos}, thus we shall formally state a noiseless coding theorem in view of  $H_\pm(X)$.

\begin{theorem}{Generalised noiseless coding theorem.}\label{t1} Let the Nagumo-Kolmogorov function $\varphi(x)=x$, then the expected lengths $L_{\pm}=\sum_{x\in X} p(x) \vert x \vert_\pm$ satisfy
\begin{equation*}
L_{\pm}\geq H_{\pm}(X),
\end{equation*} 
with equality iff $|x|^*_\pm=-\ln_{\pm}p(x)$ for every $x$ in $X$.
\end{theorem}

\begin{proof} Note that the difference
\begin{equation}
\begin{split}
L_{\pm}-H_{\pm}(X) &= \sum_{x\in X} p(x) |x|_\pm + \sum_{x\in X} p(x) \ln_{\pm}p(x)\\
& = \sum_{x\in X} p(x)\left(|x|_\pm + \ln_{\pm}p(x)\right)\\
& \geq 0,
\end{split}
\end{equation}
directly implies that $|x|_\pm \geq -\ln_{\pm} p(x)$, for the reason that every $|x|_\pm$ is an integer. Hence, the equality is attained iff the individual lengths $|x|_\pm=|x|_\pm^*$ are optimal.
\end{proof}

What the Theorem \ref{t1} tells us is that the minimum rate of data compression that can be accomplished by a codification process is bounded from below by the entropy scale that characterises the statistics of the system involved. This result has been deeply studied in \cite{noiseless}, although it conforms a cornerstone for our current purposes, namely for the statement of the following theorem.

\begin{theorem}\label{t2} Let $p(x)$ be a recursive probability distribution. For a linear cost function $\varphi(x)=x$, the entropy measures $H_\pm(X)$ induce the existence of effective complexities $K_+(X)$ and $K_-(X)$ such that
\begin{equation*}
0\leq \sum_{x\in X}p(x)K_\pm(x) - H_\pm(X) \leq K_\pm(X),
\end{equation*}
where $K_+(X)$ and $K_-(X)$ are interpreted, respectively, as average lower and upper bounds on the statistical complexity K(X).
\end{theorem}

\begin{proof} The inequality at the left implies that $H_\pm(X)\leq\sum_{x\in X}p(x)K(x)$, which is assured by Theorem \ref{t1} given that $K_\pm(x) \sim |x|$, attaining the equality as long as $K_\pm(x) = |x|_\pm^*$. On the other hand, to prove the second inequality suppose that $K_\pm(x)=c' |x|_\pm^*$, $c' \geq 1$, namely $K_\pm(x)+O(1)=|x|_\pm^*$, then 
\begin{equation*}
c' \sum_{x\in X} p(x) |x|_\pm^* + \sum_{x\in X} p(x)\ln_\pm p(x) \leq c' |X|_\pm^*,
\end{equation*}
regrouping terms on both sides, we get
\begin{equation*}
c' \sum_{x\in X} \left(1-p(x)^{-1}\right)p(x)\ln_\pm p(x) \leq \sum_{x\in X} p(x)\ln_\pm p(x),
\end{equation*}
since $p(x)<1$, it follows that $1-p(x)^{-1}<0$, hence the inequality is true, while the equality holds for $p(x)=1$. Therefore we have the theorem.
\end{proof}

As a remark, here the effective $K_\pm(X)$ are not combinatorial measures of information but statistical CMs and they shall not be directly interpreted as descriptive bounds on the Kolmogorov complexity as stated in Def. \ref{kolmogorov}, rather what the measures $K_\pm(X)$ quantify is an average rate of complexity in agreement with the information measured by the entropies $H_\pm(X)$.

As an example, consider the probability distribution:
\begin{equation*}
p(x) = \left\{ \begin{array}{ll}
					0.y & \text{if }~ x = x_1  \\
					1-0.y & \text{if }~ x = x_2 \\
					0 & \text{otherwise},
				\end{array}
			\right.
\end{equation*}
where $y$ is the binary representation of a number $0.y$ between 0 and 1. 

For the entropy measure $H_+(X)$ we have:
\begin{equation*}
0 \leq (c'-1)(-0.y \ln_+0.y - (1-0.y)\ln_+(1-0.y)) \leq c'(-\ln_+0.y - \ln_+(1-0.y)),
\end{equation*}
but $-x \ln_\pm x \leq - \ln_\pm x$, with equality if $x=1$, then from the expression above we get
\begin{equation*}
(c'-1)(-\ln_+0.y - \ln_+(1-0.y)) \leq c'(-\ln_+0.y - \ln_+(1-0.y)),\end{equation*}
analogously for $H_-(X)$. 

%	----------------------------------------------------------------------/
%	//////////////////////////////////////////////////////////////////////
%	----------------------------------------------------------------------/
\subsection{Kullback-Leibler divergence as a complexity measure}

Sometimes it might be of interest how a given probability distribution is different from another one, typically a prior with respect to a trial distribution. In the context of information this leads to define the entropy of the distribution $p$ relative to another distribution $q$, such that
\begin{equation}
\label{kullback}
H(p\parallel q)=- \sum_{x \in X} p(x)\Lambda(p(x)) + \sum_{x \in X} p(x)\Lambda(q(x)),
\end{equation}
which is a generalisation of the Kullback-Leibler divergence \cite{kullback}. 

The expression in \eqref{kullback} can be interpreted as a measure of information gain \cite{solomonoff1,solomonoff2}. This is fairly intuitive since $q$ is known as the {\it prior} in the Bayesian probability theory and conveys all the initial speculations about something before performing any observation. Yet the prior may haul redundancies, for example if $q(x)=1/{\text{dim}(X)}$ for all $x$, then we are led again to the entropy  up to a constant.

Nonetheless, as shown in \cite{baez}, the prior distribution $q$ can induce interesting results. Imagine that the prefix Turing machine $U$ runs the program $x$, delivers the outcome $y$ and halts, which is expressed in symbols as $U(x)=y$. In this respect, we arrive to a generalisation of Eq. \eqref{chaitin}: 
\begin{equation}
\label{prior}
q(y) = \left\{ \sum_{x\in X} \epsilon^{-\beta |x|}: U(x)=y \right\}, 
\end{equation}
acting as a mirror of $p\sim \epsilon^{-\beta|x|}$ over the set $\mathbb{N}$. 

Indeed, using this prior, and following the same arguments given in \cite{baez}, we are to show that one could think of \eqref{kullback} as a generalisation of a superstatistical algorithmic entropy.

Suppose now that we are particularly interested in those programs whose output is the string $s$, hence the auxiliar distribution that allows us to select that sort of programs is
\begin{equation}
p_y(s) = \left\{ \begin{array}{ll}
					1 & \text{if } y = s  \\
					0 & \text{otherwise},
				\end{array}
			\right.
\end{equation}
then we compute the entropy of $p_y(s)$ relative to \eqref{prior} to obtain
\begin{equation}
\label{infgain}
\begin{split}
H(p_y\parallel q) &=-\sum_{l\in\mathbb{N}} p_y(l)\Lambda(p_y(l)) + \sum_{l\in\mathbb{N}} p_y(l)\Lambda \left(\sum_{x\in X} \epsilon^{-\beta |x| }: U(x)=y \right) \\
&=  K_U(y) + \Lambda\left( Z \right),
\end{split}
\end{equation}
where $Z = \sum_{x\in X}\epsilon^{-\beta |x|}$ is the partition function, and the algorithmic entropy reads
\begin{equation}
\label{algorithmic}
K_U(y)=-\Lambda\left(\sum_{x\in X} \epsilon^{-\beta |x|}: U(x)=y \right),
\end{equation}
implying  that the relative entropy \eqref{infgain} is a generalisation of the algorithmic entropy in Def. \ref{kolmogorov}, cf. \cite{chaitin,zvonkin,chaitin2}.

For instance, when Eq. \eqref{infgain} is put into terms of $\Lambda(x) = \ln(x)$ and $\epsilon^x=e^x$, one can obtain the algorithmic entropy (parallel to Shannon's entropy) reported by Baez \cite{baez}, namely
\begin{equation*}
K_U(y) = -\ln\left(\sum_{x\in X} e^{-\beta |x|}: U(x)=y \right),
\end{equation*}
please note that according with Def. \ref{kolmogorov}, the complexity of the shortest program $x\in X$ that prints $y$ and halts, is specially derived from $H(y)$ for $\beta=\ln2$, yielding $K(y) = |x| + O(1)$.

The expressions \eqref{infgain} and \eqref{algorithmic} can now be inserted into the  structure of the generalised entropies $H_+(X)$ and $H_-(X)$. To simplify the analysis we are to use their series representations \eqref{series}. In that regard the entropy $H_+(X)$ of a distribution $p_y(s)$ relative to a prior $q(y)$, as defined in \eqref{prior}, becomes
\begin{equation}
\label{alg1}
\begin{split}
H_+(p_y\parallel q) &= -\sum_{l\in\mathbb{N}}\sum_{k\in\mathbb{N}} \frac{1}{k!}\left[p_y(l)\ln p_y(l)\right]^k \\
&\hspace{2.5cm} + \sum_{l\in\mathbb{N}}\sum_{k\in\mathbb{N}}\frac{1}{k!}\left[p_y(l)\ln\left(\sum\limits_{x\in X} e^{-\beta |x|}: U(x)=y\right)\right]^k\\
&= -\sum_{k\in\mathbb{N}}\frac{1}{k!}\ln^k\left(\sum_{x\in X} e^{-\beta\vert x\vert}:U(x)=y\right) \\
&\hspace{3.5cm} + \sum_{k\in\mathbb{N}}\frac{1}{k!}\ln^k\left(\sum\limits_{x\in X} e^{-\beta|x|}\right) \\
&=K_+(y) +  \sum_{k\in\mathbb{N}} \frac{1}{k!}\ln^k(Z),
\end{split}
\end{equation}
where $K_+(y)$ is the effective algorithmic entropy that generalises the one given in \cite{baez} and consequently in \cite{chaitin,zvonkin,chaitin2}.

Ditto, now we compute the entropy $H_-(X)$ of a distribution $p_y(s)$ relative to a prior $q(y)$, which yields
\begin{equation}
\label{alg2}
\begin{split}
H_-(p_y\parallel q) &= -\sum_{l\in\mathbb{N}}\sum_{k\in\mathbb{N}}\frac{(-1)^{k+1}}{k!} \left[p_y(l)\ln p_y(l)\right]^k \\
&\hspace{2.5cm} + \sum_{l\in\mathbb{N}}\sum_{k\in\mathbb{N}}\frac{(-1)^{k+1}}{k!}\left[p_y(l)\ln\left(\sum\limits_{x\in X} e^{-\beta |x|}: U(x)=y\right)\right]^k\\
&= -\sum_{k\in\mathbb{N}}\frac{(-1)^{k+1}}{k!}\ln^k\left(\sum_{x\in X} e^{-\beta\vert x\vert}:U(x)=y\right) \\
&\hspace{3.5cm} + \sum_{k\in\mathbb{N}}\frac{(-1)^{k+1}}{k!}\ln^k\left(\sum\limits_{x\in X} e^{-\beta|x|}\right) \\
&=K_-(y) +  \sum_{k\in\mathbb{N}} \frac{(-1)^{k+1}}{k!}\ln^k(Z),
\end{split}
\end{equation}
where $K_-(y)$ is the effective algorithmic entropy related to the entropy $H_-(X)$.

There is a curious aspect that we would like to highlight regarding the effective algorithmic entropies $K_+$ and $K_-$, derived from \eqref{alg1} and \eqref{alg2} respectively. Indeed they do not only generalise the algorithmic entropy in \cite{baez}, but at the same time are special cases of the relative entropies $H_+(p_y\parallel q)$ and $H_-(p_y\parallel q)$ and satisfy Theorem \ref{t2}. Suggesting that the entropies \ref{masmenos} are fundamental measures of information.

Finally, we have discussed that the algorithmic entropies $K_+$ and $K_-$ differ from the standard case $K$ in a regime of low-density programs whereas the three measures coincide asymptotically for bigger chunks of data. In principle, these functionals are uncomputable, still let us make use of a numerical trick to give some hint of their behaviour. To illustrate this, we have generated the Fig. \ref{algEntropies}, where one can observe that there is a region in which the algorithmic entropy $K_-$ would account for a more economical description than the two other measures, yet as $|x|$ (usually measured in bytes) grows the three measures tend to coincide. We cannot assess the definite impact of their actual differences on the description of complex objects, but can the complex structures whose number of components are reasonably small, have a different description than the one estimated by the standard theory?

\begin{figure}[h!]
\begin{center}
\includegraphics[width=\textwidth]{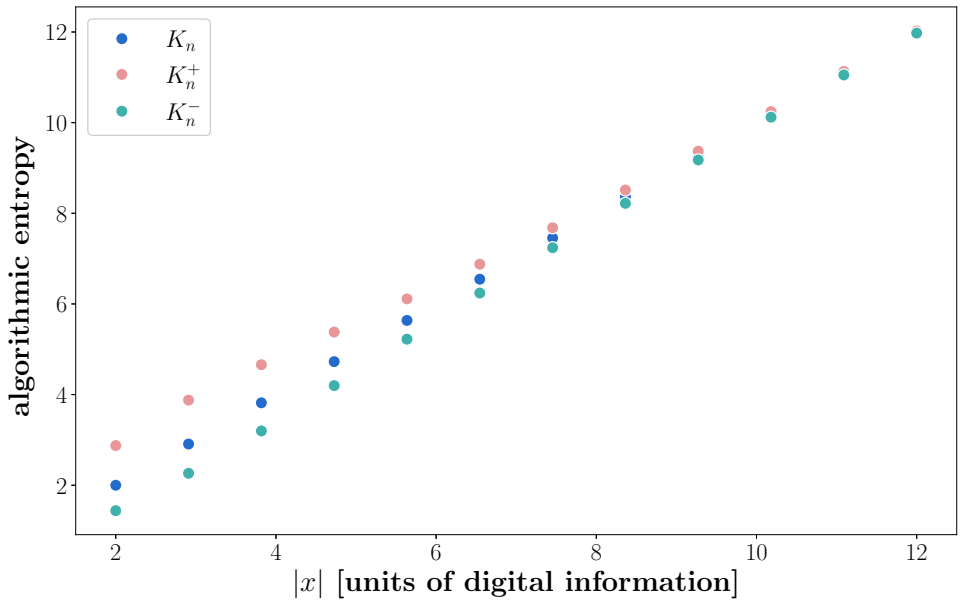} 
\caption{Numerical comparison of effective algorithmic entropies, $K_n$, $K_n^+$ and $K_n^-$ in different data size regimes $|x|$. Note how the three measures converge as $|x|$ grows, while they differ from each other at extremely small chunks of information, below 8 bits. As of today, still there are applications depending on codes of such lengths. That is the case of the American Standard Code for Information Interchange (ASCII), which is an 8-bit code.}
\label{algEntropies}
\end{center}
\end{figure}

%	----------------------------------------------------------------------/
%	//////////////////////////////////////////////////////////////////////
%	----------------------------------------------------------------------/
\section{Conclusions}
\label{conclusions}

We have explored a route that may establish a link between the algorithmic information theory and the superstatistics framework. The consequences of such unification could derive in a different understanding on the processes of information. Some of these aspects are manifested in our results that we shall summarise as follows.

Formally speaking, the connection between the superstatistical framework and the Kolmogorov complexity comes immediately by following the relation \eqref{theorem}, which is nothing but a generalisation of the statement formulated in \cite{chaitin,zvonkin,chaitin2}, which associates the Shannon's entropy with the SCM. Yet, from our generalised statistical viewpoint, any measure of complexity can be interpreted as an average rate of data compression that usually will differ from the one appraised by the standard theory.

It does not mean, even though, that all CMs are truly realisable. Note that the effective Boltzmann factor \eqref{BE} is computed from a probability distribution that may convey free parameters. Although this parametric structure grants the entropies \eqref{entropic} of enough flexibility to be potentially adapted to any circumstance. But there is a tradeoff with the functional's stability, which is a condition that shall indistinctly be satisfied in order to qualify as an information measure and, in consequence, as a CM.

It has been shown in \cite{paper} that the generalised entropies $H_\pm(X)$  absolutely fulfil the criteria of stability as the Shannon's entropy does \cite{lesche}. Besides, these are entropies that do not depend on free parameters and resemble the Shannon's entropy in the regime of high-density states. Both formulations state individually a coding theorem (see Theorem \ref{t1} and \cite{noiseless}) and also correspond to the effective CMs $K_\pm(X)$, as assured by Theorem \ref{t2}.  Even more, what Theorem \ref{t2} confirms is that the CMs $K_\pm(X)$ inherit the attributes adjudicated to the entropies $H_\pm(X)$, therefore they qualify genuinely as CMs of information.

To reinforce our arguments, we have also showed that pursuing a generalisation of the notion of relative entropy \eqref{kullback}, permits the reconstruction of the complexities $K_\pm(X)$, as represented in Eqs. \eqref{alg1} and \eqref{alg2}. These results are equivalent to our previous computations via Theorem \ref{t2}.

As a final remark, the complexities $K_\pm(X)$ converge asymptotically to the usual Kolmogorov measure, while they reflect differences with the standard theory in case that the amount of information is handled in small chunks of data. Certainly $K_-$ could imply that, when dealing with short strings, the average rate of complexity is susceptible to further compression (or a briefer description) than the one specified by $K$. We point out that there is a possible route to generalise the algorithmic information theory on the basis of superstatistics, where $H_\pm(X)$ are the unique entropies resembling the Shannon's theory while accounting for non-equilibrium phenomena without using parameters.

%	----------------------------------------------------------------------/
%	//////////////////////////////////////////////////////////////////////
%	----------------------------------------------------------------------/
\appendix
\section{Generalised logarithms and stretched exponentials}
\label{app}

Let $\epsilon:\mathbb{R}\to\mathbb{R}$ be a stretched exponential. The functions $\Lambda:\mathbb{R}\to\mathbb{R}$ fulfilling the conditions $\Lambda(1)=0$ and $\Lambda'(1)=1$, such that $\Lambda(\epsilon^x) = \epsilon^{\Lambda(x)}=x$ are called generalised (or effective) logarithms. Their series representation can usually be put into terms of the fundamental logarithm functions $\ln$ or $\log$. For a more detailed discussion than the one presented here, see Ref. \cite{hanellog}. We limit to present the basic structure of the effective logarithms $\ln_+$ and $\ln_-$, and their inverses.

We have the functions
\begin{equation}
\label{logs}
\begin{split}
\ln_{+}(x) & \equiv-\frac{1-x^x}{x}\\
\ln_{-}(x) & \equiv -\frac{x^{-x}-1}{x},
\end{split}
\end{equation}
for $x\in[0,1]$, otherwise the functions become undefined. From such definitions it becomes evident that the functions $\ln^{(\pm)}$ do not fulfil the three laws of logarithms. Yet they can be expanded in series as
\begin{equation}
\ln_+(x) = \ln x + \frac{1}{2!} x \ln^2x + \frac{1}{3!} x^2\ln^3x + \frac{1}{4!} x^3\ln^4x + \cdots,
\end{equation}
and
\begin{equation}
\ln_-(x) = \ln x - \frac{1}{2!} x \ln^2x + \frac{1}{3!} x^2\ln^3x - \frac{1}{4!} x^3\ln^4x + \cdots,
\end{equation}
note that the first term is in both cases leads the series, while higher order terms become subdominant as $x\to0$. This peculiar flexibility grants to entropies \eqref{masmenos} the simultaneous character of accounting for non-equilibrium phenomena in the low-probability regime, while preserving a well defined thermodynamical limit.

The corresponding stretched exponentials of \eqref{logs} do not posses a closed form, in this case we make use of a numerical representation. These functions have been constructed as
\begin{equation}
\label{exps}
\exp_{\pm}(x)\equiv\exp(-x)\sum_{j=0}^\infty a_{\pm}(j) x^{j}, \quad a_\pm(j) \in \mathbb{R},
\end{equation}
the first nine coefficients $a_\pm(j)$ are given in Table \ref{t:apm}.

\vspace{0.25cm}
\renewcommand{\arraystretch}{1.4}
\begin{table}[!htp]
\centering
\begin{tabular}{r>{\raggedleft}p{0.25\linewidth}>{\raggedleft\arraybackslash}p{0.25\linewidth}}
\hline
& $a_{+}(j)$ & $a_{-}(j)$ \\ \hline
$j=8$ & -0.000157095 & 0.000105402 \\
$j=7$ & 0.00373467   & -0.00211934 \\
$j=6$ & -0.0362676   & 0.0166679 \\
$j=5$ & 0.186358     & -0.0675544 \\
$j=4$ & -0.546751    & 0.16867 \\
$j=3$ & 0.905157     & -0.317048 \\
$j=2$ & -0.709322    & 0.3725 \\
$j=1$ & 0.0228963    & 0.0147449 \\
$j=0$ & 1 & 1 \\ \hline
\end{tabular}
\caption{$a^{(\pm)}$ coefficients.}
\label{t:apm}
\end{table}

%	----------------------------------------------------------------------/
%	//////////////////////////////////////////////////////////////////////
%	----------------------------------------------------------------------/
\bibliographystyle{unsrt}
\bibliography{mybib}
%	----------------------------------------------------------------------/
\end{document}